\documentclass[submission,copyright,creativecommons]{eptcs}
\providecommand{\event}{AFL 2017} 
\usepackage{underscore}           
\usepackage{amssymb} 
\usepackage{latexsym}
\usepackage{amsmath}
\usepackage{amsthm}
\usepackage{verbatim}
\usepackage{algorithm}
\usepackage[noend]{algorithmic}
\usepackage{graphicx}
\usepackage{xcolor}
\usepackage{epsfig}
\usepackage{tikz}

\tikzstyle{vertex}=[circle, thick, draw=black, inner sep = 2pt]
\tikzstyle{positive}=[thick]

\newcommand{\abequiv}{
        \raisebox{-2pt}{
        \begin{tikzpicture}
        \draw[->] (0,-0.1) -- (0,0.22);
        \draw (0, 0.06) circle(0.05);
        \end{tikzpicture}}\,\,}

\newcommand{\compat}{
        \raisebox{-2pt}{
        \begin{tikzpicture}
        \draw[->] (0,-0.1) -- (0,0.22);
        \end{tikzpicture}}\,\,}

\newcommand{\Z}{\mathbb{Z}}
\newcommand{\D}{\mbox{D}}

\renewcommand{\d}{\ensuremath{\diamond}}

\newtheorem{lemma}{Lemma}
\newtheorem{proposition}{Proposition}
\newtheorem{theorem}{Theorem}
\newtheorem{corollary}{Corollary}
\newtheorem{conjecture}{Conjecture}

\newtheorem{defn}{Definition}

\title{Dyck Words, Lattice Paths, and Abelian Borders}
\author{F. Blanchet-Sadri \qquad\qquad Kun Chen
\institute{Department of Computer Science\\
University of North Carolina\\
P.O. Box 26170\\
Greensboro, NC 27402-6170, USA}
\email{blanchet@uncg.edu \qquad\qquad k\_chen2@uncg.edu}
\and
Kenneth Hawes
\institute{Department of Mathematics\\ 
University of Virginia\\
P.O. Box 400137\\
Charlottesville, VA 22904-4137, USA}
\email{kgh5cc@virginia.edu}
}

\def\titlerunning{Dyck Words, Lattice Paths, and Abelian Borders}
\def\authorrunning{F.~Blanchet-Sadri, K.~Chen \& K.~Hawes}

\begin{document}

\maketitle

\begin{abstract}
We use results on Dyck words and lattice paths to derive a formula for the exact number of binary words of a given length with a given minimal abelian border length, tightening a bound on that number from Christodoulakis et al.\ (Discrete Applied Mathematics, 2014). We also extend  to any number of distinct abelian borders a result of Rampersad et al.\ (Developments in Language Theory, 2013) on the exact number of binary words of a given length with no abelian borders. Furthermore, we generalize these results to partial words.


\end{abstract}

\section{Introduction}

{\em Abelian borders} and {\em abelian periods} are defined similarly to their classical counterparts based on abelian equivalence. Two words are {\em abelian equivalent} if they have the same {\em Parikh vector} (a vector that records the frequency of each letter occurring in the word); equivalently, if the two words are the same up to a permutation of characters.  For example, the word $010000010010100$ has, in particular, an abelian border of length 3 since the prefix of length 3, $010$, is abelian equivalent to the suffix of length 3, $100$. It has an abelian period of length 4 as the following factorization suggests: $0100\mid 0001\mid 0010\mid 100$. Connections between these concepts exist (see, e.g., \cite{CHPZ}). Abelian versions of some well-known results such as Fine and Wilf's periodicity theorem \cite{CoIl06,BSSiTeVe} and the critical factorization theorem \cite{AvKaPu} were proposed. Several algorithms that efficiently compute abelian periods in a given word were developed \cite{CrIlKoKuPaRaRyTyWa,FiLeLePG,KoRaRy}. Abelian concepts were reviewed in \cite{ChCh}; applications include sequencing from compomers \cite{Boc04}, permutation pattern discovery in biosequences \cite{ErLaPa}, gapped permutation patterns for comparative genomics \cite{Par}, and jumbled pattern matching \cite{BuCiFiLi1,BuCiFiLi2,BuErLa}. 

Christodoulakis et al.\  \cite{ChChCrIl} studied abelian borders in words over an alphabet of size two. They investigated the number of such words with fixed minimal border length by establishing connections with Dyck words. A \textit{Dyck word} is a word over the alphabet $\{0,1\}$ that has the same number of zeros and ones and that has no prefix with more ones than zeros, e.g., 0100110101 is a Dyck word but 0110 is not. The language of Dyck words is equivalent to the language of balanced parentheses: 0101000111 corresponds to $()()((()))$. The number of Dyck words of length $2n$ is given by the Catalan number $C_n=\frac{1}{n+1}{{2n}\choose{n}}$. Christodoulakis et al.\ \cite{ChChCrIl} gave bounds on the number of binary words of a given length with no abelian borders and on the number of binary words of a given length with at least one abelian border. They also described an algorithm to find the minimal abelian border of a binary word of length $n$ in $\Theta(\sqrt{n})$ time on average (when the word has an abelian border). They left as open problems the derivation of tighter bounds and the generalization to larger alphabet sizes.

Rampersad et al.\ \cite{RaRiSa} investigated the number of words with or without abelian borders by establishing connections with {\em lattice paths}; in particular, they used lattice paths to find the exact number of binary words of a given length with no abelian borders. Similarly to a construction given in \cite{RaRiSa}, we can visualize a word $w$ over $\{0,1\}$ as an ordered pair of increasing lattice paths, where one path corresponds to the prefixes of the word, and the other corresponds to the suffixes. A step to the right represents the letter 0, and a step upward represents the letter 1. An intersection between the two paths after ${k}$ steps indicates an abelian border of length ${k}$. If we graph the distance in steps between the prefix path and the suffix path, we get a \textit{Motzkin path} which consists of steps diagonally upward corresponding to 0, steps diagonally downward corresponding to 1, and steps straight forward corresponding to 2. If we remove the straight forward steps, we get a \textit{Dyck path}. A binary word has an abelian border of a given length if and only if the corresponding Dyck path is encoded by a Dyck word or the bitwise negation of a Dyck word \cite{ChChCrIl}. 

Since lattice paths have been studied independently of the (relatively new) concept of abelian borders (see, e.g., \cite{Deutsch,GeGoShWiYe}), in this paper we use known results about lattice paths to discover properties of words with abelian borders. Among other results, by counting pairs of lattice paths by intersections, we derive a formula, that involves the Catalan number, for the exact number of binary words of a given length with a given minimal abelian border length, tightening the abovementioned bound on that number from Christodoulakis et al.\ \cite{ChChCrIl}. We also extend the abovementioned result of Rampersad et al.\ \cite{RaRiSa} to any number of distinct abelian borders. Furthermore, we generalize these results to partial words. Other enumerative applications of non-intersecting lattice paths can be found in \cite{GeVi}. 

Our paper is the first attempt to study abelian borders of partial words (abelian periods were studied in \cite{BSSiTeVe}). Our paper's contents are as follows: 
In Section~2, we review some basic concepts on partial words such as abelian borders. 
In Section~3, we describe an $O(n)$ time algorithm to compute the minimal abelian border of a given partial word of length $n$ over an alphabet of any size. It is an adaptation of an algorithm from \cite{ChChCrIl}. 
In Section~4, we extend to partial words a relation from \cite{ChChCrIl} between minimal abelian borders and Dyck words.
In Section~5, we apply results from intersecting lattice paths to the enumeration of abelian borders for total words. We also count abelian bordered partial words using lattice paths.  
Finally in Section~6, we conclude with some remarks. 

\section{Preliminaries}

Let $\Sigma$ be a finite and non-empty set of characters, called an \emph{alphabet}.  A sequence of characters from $\Sigma$ is referred to as a \emph{total word} over $\Sigma$. A \emph{partial word} over $\Sigma$ is a sequence of characters from $\Sigma_{\d}=\Sigma\cup\{\d\}$, where $\d$, a new character which is not in $\Sigma$, is the ``don't care'' or ``hole'' character (it represents an undefined position). A partial word may have zero or more holes, while a total word has zero holes. A total word is also a partial word. In the rest of the paper, when we refer to a ``letter'' we mean a character from the alphabet $\Sigma$ (we will never call the $\d$ character a letter since it does not belong to $\Sigma$). We will use ``word'' and ``total word'' as equivalent terms. The \emph{length} of a partial word $w$, denoted by $|w|$, is the number of $\Sigma_{\d}$ characters in $w$, e.g., if $w=abbac{\d}{\d}c$ then $|w|=8$. The \emph{empty word} is the word of length zero and we let $\varepsilon$ denote it. The set of all words over $\Sigma$ is denoted by $\Sigma^{*}$.  Similarly, the set of all non-empty words over $\Sigma$ is denoted by $\Sigma^{+}$.  We let $\Sigma^{n}$ denote the set of all words of length $n$ over $\Sigma$. We can similarly define $\Sigma_{\d}^{*}$, $\Sigma_{\d}^{+}$, and $\Sigma_{\d}^{n}$ for partial words over $\Sigma$.  

A \emph{completion} over $\Sigma$ of a partial word $w$, denoted by $\hat{w}$, is a total word constructed by replacing each $\d$ in $w$ with some letter in $\Sigma$ (i.e., we fill all the holes in $w$). If $w$ has one or more holes and $\Sigma$ has more than one letter, then there is more than one distinct completion of $w$ and in this case $\hat{w}$ denotes an arbitrary completion of $w$ unless otherwise specified. 

A partial word $u$ is a {\it factor} of a partial word $w$ if there exist (possibly empty) partial words $x,y$ such that $w=xuy$.  We say that $u$ is a {\it prefix} of $w$ if $x=\varepsilon$.  Similarly, $u$ is a {\it suffix} of $w$ if $y=\varepsilon$.  Starting numbering positions from $0$, we let $w[i]$ denote the character in position~$i$ of $w$ and let $w[i..j]$ denote the factor of $w$ from position~$i$ to position~$j$ (inclusive). We let $w^m$ denote $w$ concatenated with itself $m$ times.

If $u$ and $v$ are partial words of equal length over $\Sigma$, then $u$ is \emph{contained} in $v$, denoted by $u\subset v$, if $u[i]=v[i]$ for all $i$ such that $u[i]\in\Sigma$; in other words, the set of non-hole positions in $u$ is a subset of the set of non-hole positions in $v$. Partial words $u$ and $v$ are \emph{compatible}, denoted by $u\compat v$, if there exists a partial word $w$ such that $u\subset w$ and $v\subset w$, i.e., there are completions $\hat{u}$ of $u$ and $\hat{v}$ of $v$ such that $\hat{u}=\hat{v}$.

The number of occurrences of a letter $a$ in a partial word $w$ is denoted by $|w|_a$. The \emph{Parikh vector} of $w$ over $\Sigma=\{a_0, \ldots, a_{k-1}\}$ is defined as $\psi(w)=\langle|w|_{a_0},\ldots,|w|_{a_{k-1}}\rangle$. Note that we do not count occurrences of the ${\d}$ character in Parikh vectors.  A partial word $u$ is \emph{abelian compatible} with a partial word $v$, denoted by $u\abequiv v$, if a permutation of $u$ is compatible with $v$ (this implies that $u$ and $v$ are of equal length). For example, the partial words $ab{\d}bb{\d}$ and $bbb{\d}ab$ are abelian compatible. Note that for total words $u$ and $v$ of same length, we have $u\abequiv v$ if and only if $\psi(u)=\psi(v)$, i.e., if $u$ and $v$ are permutations of one another (we also say that $u$ and $v$ are \emph{abelian equivalent}).

For a non-empty partial word $w$, if non-empty partial words $x_1, x_2, u, v$ exist such that $w = x_1 u = vx_2$ and $x_1 \abequiv x_2$, we call  $w$ {\em abelian bordered}.  In this case, a total word $x$ exists such that $x_1 \abequiv x$ and $x_2 \abequiv x$, and $x$ is called an {\em abelian border} of $w$ (we also call $x_1, x_2$ abelian borders of $w$). Two abelian borders $x, y$ of $w$ are \emph{equivalent} if and only if $|x|=|y|$. We refer to non-equivalent abelian borders as \emph{distinct}. An abelian border $x$ of $w$ is {\em minimal} if $|x| > |y|$ implies that $y$ is not an abelian border of $w$ and $x$ is {\em maximal} if $|x| < |y|$ implies that $y$ is not an abelian border of $w$.  

Equivalence of abelian borders being an equivalence relation,  we identify abelian borders by their equivalence classes to avoid counting equivalent borders multiple times, but we will not refer to the equivalence classes explicitly. Since we are only interested in identifying and counting distinct abelian borders, we hereafter use the phrase \emph{the abelian border of length} $k$ to mean the equivalence class of abelian borders of length $k$, and we refer to this equivalence class by one of its representatives.  

If a total word of length $n$ has an abelian border of length $k$, then it must also have an abelian border of length $n-k$. For example, the word $babbbbba$ of length 8 has abelian borders of lengths 2, 3, 4, 5, and 6.  Note that equivalent abelian borders of a total word always come in pairs of the same length.  Note also that for a total word of length $n$, abelian borders come in complementary pairs whose lengths sum to $n$ (if $n$ is even, then any abelian border of length $\frac{n}{2}$ is equivalent to its complementary partner) \cite{ChChCrIl}.  However for a partial word of length $n$, we may have an abelian border of length greater than $\left\lfloor\frac{n}{2}\right\rfloor$ with no complementary abelian border of shorter length. For example, $bb{\d}a$ has no abelian borders of length two or less, but it does have the non-abelian border $bba$ ($bb{\d} \abequiv bba$ and $b{\d}a \abequiv bba$). 

The following proposition shows that all partial words of length at least two without abelian borders are total words.

\begin{proposition}
If $w$ is a partial word of length greater than one, with at least one hole, then $w$ has an abelian border.
\label{nonborderedpartials}
\end{proposition}

\section{Computing the minimal abelian border}

We now present an algorithm, based on algorithm \textsc{Shortest-Abelian-Border} in \cite{ChChCrIl}. Our algorithm computes the length of the minimal abelian border of a given non-empty partial word $w$ (over an alphabet of any size) and runs in $O(n)$ time, where $n$ is the length of the input $w$. We provide the pseudo-code.

\begin{algorithm}
\caption{\textsc{Minimal-Abelian-Border-Partial}($w,n,\sigma$)}
\label{minimalborder}
\begin{algorithmic}[1]
\REQUIRE $w$, a partial word with letters in $\Sigma=\{0,1, \ldots,\sigma-1\}$,  and $n\geq1$, the length of $w$

\STATE $V \gets \langle0,0,\ldots,0\rangle$, $|V| = 0$, and $holes = 0$;

\IF{$w[0] = \d$ $\mathbf{or}$ $w[n-1] = \d$} 
	\RETURN 1;
	\ELSE 
		\STATE $V[w[0]]=V[w[0]]+1$ and $V[w[n-1]]=V[w[n-1]]-1$;
		\IF{$V[w[0]]$ = 0}
			\RETURN 1;
		\ELSE
			\STATE $|V| = 2$;
		\ENDIF
		\FOR{$i\gets1$ \TO $n-1$}
			\IF{$w[i]=\d$}
				\STATE $holes=holes+1$;
			\ELSE
				\STATE $V[w[i]]=V[w[i]]+1$;
				\IF{$V[w[i]]>0$}
					\STATE $|V|=|V|+1$;
				\ELSIF{$V[w[i]]\leq0$}
					\STATE $|V|=|V|-1$;
				\ENDIF
			\ENDIF
			\IF{$w[n-1-i]=\d$}
				\STATE $holes=holes+1$;
			\ELSE
				\STATE $V[w[n-1-i]]=V[w[n-1-i]]-1$;
				\IF{$V[w[n-1-i]]\geq0$}
					\STATE $|V|=|V|-1$;
				\ELSIF{$V[w[n-1-i]]<0$}
					\STATE $|V|=|V|+1$;
				\ENDIF
			\ENDIF
			\IF{$|V|\leq holes$}
				\RETURN $i+1$;
			\ENDIF
		\ENDFOR
		\RETURN $n$;
\ENDIF

\end{algorithmic}
\end{algorithm}

Algorithm \textsc{Minimal-Abelian-Border-Partial} works as follows. We can assume without loss of generality that $w$ is a partial word over the alphabet $\{0,\ldots,{\sigma-1}\}$. We define $V$ as a vector which gives the difference between the Parikh vectors of a prefix and a suffix of equal length, as defined in \cite{ChChCrIl}. We define the magnitude of $V$, denoted by $|V|$, as the sum of the absolute values of the entries in $V$. This variable $|V|$ represents the ``absolute total difference'' between two Parikh vectors. The variable $holes$ keeps track of the number of holes we have encountered so far in $w$.

Every time we read a letter from the left side of the partial word, we increment the entry of $V$ corresponding to that letter, e.g., if we read the letter $3$ then we increment $V[3]$. Similarly, every time we read a letter on the right side, we decrement the corresponding entry of $V$. This process of incrementing and decrementing corresponds to calculating Parikh vectors for the prefix and suffix and then subtracting them. Each time we increment or decrement an entry in $V$, we check the new value of that entry and adjust $|V|$ accordingly. For example, if we increment $V[w[i]]$ by $1$ and $V[w[i]] \leq 0$ after this action, then $|V|$ must be decreased by 1.  
Every time we read a $\d$ instead of a letter, we count it by incrementing $holes$ but do nothing to $V$ since holes are not counted in Parikh vectors.

Each time we finish moving ``inward'' by one position on each side, we check to see if we have found an abelian border. This is equivalent to checking whether the prefix up to the current position is abelian compatible with the suffix of the same length; this is equivalent to checking whether $|V|$ is less than or equal to the total number of holes in both the prefix and the suffix.

Unlike the algorithm in \cite{ChChCrIl} for total words, we must run through the entire length of the partial word instead of stopping halfway. If the input partial word has no abelian borders, then we return the length of the partial word. Note that by Proposition~\ref{nonborderedpartials}, this can only occur when the input is a total word.

\begin{proposition}
Algorithm \textsc{Minimal-Abelian-Border-Partial} computes the minimal abelian border for any partial word of length $n$ over any finite alphabet and runs in $O(n)$ time.
\label{brdralgthm}
\end{proposition}

\section{Abelian borders and Dyck words}

We extend a result from \cite{ChChCrIl}, which relates minimal abelian borders to Dyck words (see Proposition~\ref{dyckborder} below). We start with a lemma and a definition. 

\begin{lemma}
A partial word $w$ has an abelian border of length ${k}$ if and only if there exists a completion of $w^2$ with an abelian border of length ${k}$.
\label{squarecompletion}
\end{lemma}

\begin{defn}
Let $w$ be a total word over $\{0,1\}$ with $|w|=n\geq2$. Let $1\leq {k}\leq n$.
\begin{itemize}
\item
The $\mathcal{Y}_{k}$-form of $w$, denoted by $\mathcal{Y}_{k}(w)$, is a total word over $\{0,1,2\}$ of length $k$  defined by

\[\mathcal{Y}_{k}(w)[i]=
\begin{cases}
0, & \text{if }w[i]=0\text{ and }w[n-1-i]=1; \\
1, & \text{if }w[i]=1\text{ and }w[n-1-i]=0; \\
2, & \text{if }w[i]=w[n-1-i].
\end{cases}\]

\item
The $\mathcal{Z}_{k}$-form of $w$, denoted by $\mathcal{Z}_{k}(w)$, is a total word over $\{0,1\}$ defined by removing the twos from $\mathcal{Y}_{k}(w)$.
\end{itemize}
\label{yzform}
\end{defn}

For example, if $w=1010011001$, then $\mathcal{Y}_{5}(w)=22100$ and $\mathcal{Z}_{5}(w)=100$. If $w=010$, then $\mathcal{Y}_3(w)=222$ and $\mathcal{Z}_3(w)=\varepsilon$. Note that $\mathcal{Z}_{k}(w)$ is a subsequence of (not necessarily consecutive) letters of $w$.

A \emph{Dyck word} of length $2n\geq0$ is a total word over $\{0,1\}$ which consists of $n$ zeros and $n$ ones arranged so that no prefix of the word has more ones than zeros. We call a $\mathcal{Z}_{k}$-form \emph{Dyckian} if it is a Dyck word or the bitwise negation of a Dyck word. For a given partial word $w$ of length $n$ with $h$ holes, we let $\D(w)$ denote the set of all $\mathcal{Z}_{k}((\hat{w^2})_i)$ such that $1\leq {k}\leq n$, $0\leq i\leq2^{2h}-1$,  $(\hat{w^2})_i$ is a completion of $w^2$, and $\mathcal{Z}_{k}((\hat{w^2})_i)$ is Dyckian. Note that there are $2^{2h}$ distinct completions of $w^2$ (since $w^2$ has $2h$ holes) and they can be ordered lexicographically. For example, let $w=11{\d}0$. Then the completions of $w^2=11{\d}011{\d}0$ are
$(\hat{w^2})_0=11001100, (\hat{w^2})_1=11001110, (\hat{w^2})_2=11101100, (\hat{w^2})_3=11101110$.
So $\D(w)=\{\mathcal{Z}_4((\hat{w^2})_0),\mathcal{Z}_3((\hat{w^2})_1),\mathcal{Z}_4((\hat{w^2})_3)\}$. Looking at $\D(w)$, the following proposition implies that the minimal abelian border of $w$ has length 3. The proof's main idea comes from \cite[Lemma~1]{ChChCrIl}.
 
\begin{proposition}
A partial word $w$ over $\{0,1\}$ of length at least two has a minimal abelian border of length  
$\min\{{k}\;|\;\text{\emph{there exists }}i\text{\emph{ such that }} \mathcal{Z}_{k}((\hat{w^2})_i)\in\D(w)\}$.
\label{dyckborder}
\end{proposition}

\section{Abelian borders and lattice paths}

We can represent a word over $\{0,1\}$ as an ordered pair of paths of equal length on the lattice $\Z^2$. The first path corresponds to the prefixes, and the second path corresponds to the suffixes. Both paths begin at the origin (the southwest corner of the lattice). For the prefix, we start with the leftmost letter and step right every time we see a zero and up every time we see a one. For the suffix, we start with the rightmost letter and step in the same manner. Note that both paths are monotonically increasing, also called \emph{minimal}. 

\begin{lemma}[\cite{RaRiSa}]
A total word $w$ over $\{0,1\}$ has an abelian border of length ${k}$ if and only if the lattice paths for $w$ meet after ${k}$ steps.
\label{latticemeeting}
\end{lemma}

Figure~\ref{borderedlattice} gives an example of a lattice representation of a word that has abelian borders. If we draw the paths diagonally (taking a step diagonally upward for a zero and diagonally downward for a one) and treat them as discrete functions of the prefix (suffix) length, then we can easily take the absolute difference between these functions, and this difference is precisely $|V|$ as in Algorithm \textsc{Minimal-Abelian-Border-Partial}. Note that since we are dealing with total words, we have an abelian border of length ${k}$ if and only if $|V|=0$ after ${k}$ steps, and this corresponds to a meeting of the two paths. 

We can also think of the graph of $|V|$ as corresponding to the $\mathcal{Y}_{k}$-form of $w$, where a zero gets a step diagonally upward, a one gets a step diagonally downward, and a two gets a step straight forward. Note that this is a \emph{Motzkin path}. Similarly, we can extend this correspondence to $\mathcal{Z}_{k}$-forms by removing all of the straight steps from the graph. Just as we saw a relationship between $\mathcal{Z}_{k}$-forms, Dyck words, and abelian borders in Proposition \ref{dyckborder}, here we have that a word has an abelian border of length ${k}$ if and only if the $\mathcal{Z}_{k}$ graph is a\emph{ Dyck path}. A Dyck path of length $2n$ is a sequence of $n$ upward steps and $n$ downward steps that starts at the origin and never passes below the horizontal axis. As the name suggests, Dyck paths are coded by Dyck words \cite{Deutsch}.

\begin{figure}[h!]
\begin{center}
\begin{minipage}{.33\textwidth}
\begin{tikzpicture}
\draw[step=0.50cm,color=lightgray] (-1,-2) grid (3,2);
\draw[ultra thick, color=gray] (-1,-2) -- ++(0.5,0) -- ++(0,0.5) -- ++(0.5,0) -- ++(0.5,0) -- ++(0.5,0) -- ++(0,0.5) -- ++(0.5,0) -- ++(0,0.5);
\draw[ultra thick] (-1,-2) -- ++(0,0.5) -- ++(0.5,0) -- ++(0,0.5) -- ++(0.5,0) -- ++(0.5,0) -- ++(0.5,0) -- ++(0,0.5) -- ++(0.5,0);
\draw [->] (-1,-2) -- (-1,2);
\draw [->] (-1,-2) -- (3,-2);
\node at (1, -2.5) {Lattice representation};
\end{tikzpicture}
\end{minipage}
\begin{minipage}{.33\textwidth}
\begin{tikzpicture}
\draw[step=0.50cm,color=lightgray] (-1,-2) grid (3,2);
\draw[ultra thick, color=gray] (-1,0) -- ++(0.5,-0.5) -- ++(0.5,0.5) -- ++(0.5,-0.5) -- ++(0.5,-0.5) -- ++(0.5,-0.5) -- ++(0.5,0.5) -- ++(0.5,-0.5) -- ++(0.5,0.5);
\draw[ultra thick] (-1,0) -- ++(0.5,0.5) -- ++(0.5,-0.5) -- ++(0.5,0.5) -- ++(0.5,-0.5) -- ++(0.5,-0.5) -- ++(0.5,-0.5) -- ++(0.5,0.5) -- ++(0.5,-0.5);
\draw [->] (-1,0) -- (3,0);
\draw [<->] (-1,2) -- (-1,-2);
\node at (1, -2.5) {Diagonal form};
\end{tikzpicture}
\end{minipage}

\begin{minipage}{.33\textwidth}
\begin{tikzpicture}
\draw[step=0.50cm,color=lightgray] (-1,-2) grid (3,2);
\draw (-1,-2) -- ++(0.5,1) -- ++(0.5,-1) -- ++(0.5,1) -- ++(0.5,0) -- ++(0.5,0) -- ++(0.5,-1) -- ++(0.5,1) -- ++(0.5,-1);
\draw [->] (-1,-2) -- (-1,2);
\draw [->] (-1,-2) -- (3,-2);
\node at (1, -2.5) {$\mathcal{Y}_{8}$ graph};
\end{tikzpicture}
\end{minipage}
\begin{minipage}{.33\textwidth}
\begin{tikzpicture}
\draw[step=0.50cm,color=lightgray] (-1,-2) grid (3,2);
\draw (-1,-2) -- ++(0.5,1) -- ++(0.5,-1) -- ++(0.5,1) -- ++(0.5,-1) -- ++(0.5,1) -- ++(0.5,-1);
\draw [->] (-1,-2) -- (-1,2);
\draw [->] (-1,-2) -- (3,-2);
\node at (1, -2.5) {$\mathcal{Z}_{8}$ graph};
\end{tikzpicture}
\end{minipage}
\caption{The word $w=01000101$ has abelian borders of length 2 and 6; $\mathcal{Y}_8(w)=10122010$ and $\mathcal{Z}_8(w)=101010$.}
\label{borderedlattice}
\end{center}
\end{figure}
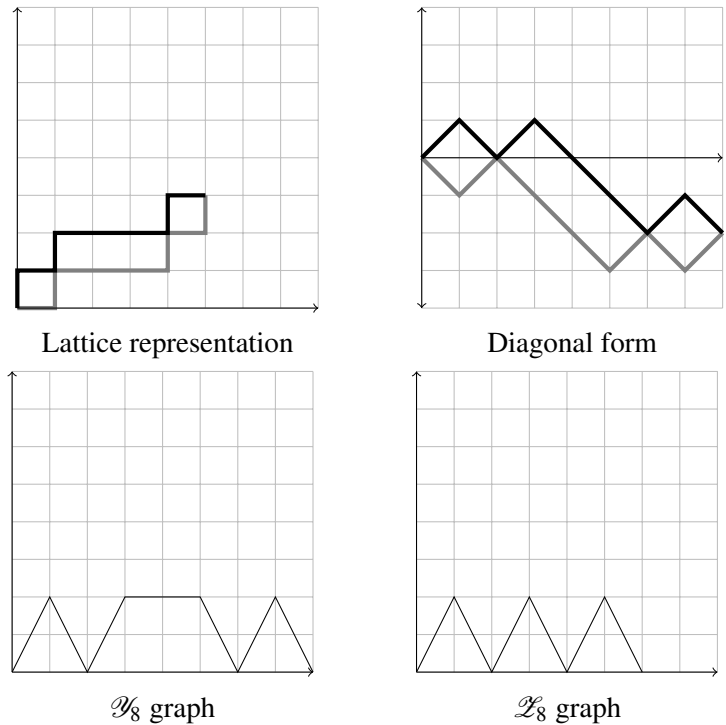

\subsection{The total word case}

First, we count binary words of a given length with a given minimal abelian border length. Theorem~4 in \cite{ChChCrIl} gives a bound on that number, but by counting pairs of lattice paths we can find the exact number. Our formula involves the \emph{Catalan number}, $C_n=\frac{1}{n+1}{{2n}\choose{n}}$, which enumerates Dyck words of length $2n$. 

\begin{theorem}
The number of binary total words of length $n$ with minimal abelian border of length ${k}$, $1\leq {k}\leq\left\lfloor\frac{n}{2}\right\rfloor$, is $2^{n-2{k}+1}C_{{k}-1}$.
\label{minimalborderfull}
\end{theorem}
\begin{proof}
The result can easily be checked for $k=1$, so suppose that $k\geq 2$. Clearly, we want to count the number of ordered pairs of increasing lattice paths of equal length which meet for the first time after ${k}$ steps. For a given rectangle with dimensions $r\times(n-r)$ on the lattice, the number of ordered pairs of increasing lattice paths which begin at the southwest corner of the rectangle and meet at the northeast corner (without meeting in between) is given by
\[2N_{n-1,r}=\frac{2}{n-1}{{n-1}\choose{r}}{{n-1}\choose{r-1}}\]
where $N_{n,r}=\frac{1}{n}{n\choose r}{n\choose r-1}$ is the \emph{Narayana number} \cite{GeGoShWiYe}.

For an abelian border of length ${k}$, we want to find all such paths for all rectangles such that the southwest corner is at the origin and the perimeter of the rectangle is $2{k}$. It is easy to see that there are ${k}-1$ such rectangles. For each rectangle, we calculate the number of permissible pairs of paths, giving a total of
\[\sum_{r=1}^{{k}-1}\frac{2}{{k}-1}{{{k}-1}\choose{r}}{{{k}-1}\choose{r-1}}.\]
There is a relationship between the Narayana numbers and the Catalan numbers, namely
\begin{equation}
\sum_{r=1}^{n}N_{n,r}=\sum_{r=1}^{n}\frac{1}{n}{n\choose r}{n\choose r-1}=\frac{1}{n+1}{{2n}\choose{n}}=C_n,\label{eq0}
\end{equation}
which can be seen by the fact that $N_{n,r}$ counts the number of Dyck paths of length $2n$ that have $r$ peaks \cite{KoSh}. Using Equation~(\ref{eq0}), we get
\[\sum_{r=1}^{{k}-1}\frac{2}{{k}-1}{{{k}-1}\choose{r}}{{{k}-1}\choose{r-1}}=2C_{{k}-1}.\]
Since we do not care what happens after the minimal abelian border, every position numbered from ${k}$ to $n-{k}-1$ can either be one of two distinct letters, so we also need to count each of these possibilities. There are $2^{n-2{k}}$ such configurations possible, so for the total number of binary words of length $n$ with minimal abelian border of length ${k}$ we arrive at $2^{n-2{k}+1}C_{{k}-1}$.
\end{proof}

Second, we can also use lattice paths to find the exact number of binary total words of a given length that have no abelian borders. This result was proved by Rampersad et al.\ in \cite{RaRiSa} using a similar method with Motzkin paths.

\begin{theorem}
The number of binary partial words of length $n$ with no abelian borders is equal to the number of binary total words of length $n$ with no abelian borders, which is
${{n}\choose{\frac{n}{2}}}$
when $n$ is even, and 
$2{{n-1}\choose{\frac{n-1}{2}}}$
when $n$ is odd.
\label{unborderedwords}
\end{theorem}

Third, we use lattice path methods to arrive at a more general formula for counting binary total words of a given length with a given number of distinct abelian borders. As was the case with Theorem~\ref{minimalborderfull}, the following theorem uses formulas for counting the number of ordered pairs of lattice paths that intersect in particular ways \cite{GeGoShWiYe}. Note that when $m=0$, this reduces to Theorem \ref{unborderedwords}.  

\begin{theorem}
The number of binary total words of length $n$ with $m$ distinct abelian borders is 
\[\begin{cases}
2^{\frac{m}{2}+1}\frac{\left(n-\frac{m}{2}-1\right)!}{\left(\frac{n}{2}\right)!\left(\frac{n-m}{2}-1\right)!}, & \text{if }n\text{ and }m\text{ are both even;} \\
2^{\frac{m-1}{2}}(m+1)\frac{\left(n-\frac{m+3}{2}\right)!}{\left(\frac{n}{2}\right)!\left(\frac{n-m-1}{2}\right)!}, & \text{if }n\text{ is even and }m\text{ is odd;} \\
2^{\frac{m}{2}+1}{{n-\frac{m}{2}-1}\choose{\frac{n-1}{2}}}, & \text{if }n\text{ is odd and }m\text{ is even;} \\
0, & \text{if }n\text{ and }m\text{ are both odd.}
\end{cases}\]
\label{binarywordswithborders}
\end{theorem}
\begin{proof}
When $n$ and $m$ are even, the number of binary words of length $n$ with $m$ distinct abelian borders is precisely the number of ordered pairs of minimal lattice paths of length $\frac{n}{2}$ which start at the origin, intersect exactly $\frac{m}{2}$ times, and end at different points. By \cite{GeGoShWiYe}, this is equal to
\begin{align*}
&\left(\#\text{ of such paths that intersect }{\frac{m}{2}}\text{ times and end anywhere}\right)\\
&-\left(\#\text{ of such paths that intersect }{\frac{m}{2}-1}\text{ times and end at the same point}\right).
\end{align*}
If $p=\frac{n}{2}$ and $q=\frac{m}{2}$, then this quantity is
\begin{equation}
2^q{{2p-q}\choose{p}}-q2^{q}\frac{(2p-q-1)!}{p!(p-q)!}.\label{eq00}
\end{equation}
By substituting our fractions into Equation~(\ref{eq00}) and simplifying, we get our result.

When $n$ is even and $m$ is odd, first recall that for total words, abelian borders come in complementary pairs whose lengths sum to $n$. The only abelian border length that does not have a complementary partner which is distinct is the abelian border of length $\frac{n}{2}$, so the only way a word can have an odd number $m$ of abelian borders is if it has an abelian border of this length. This corresponds to pairs of minimal lattice paths of length $\frac{n}{2}$ which begin at the origin, end at the same point, and intersect $\frac{m-1}{2}$ times (not including the origin or the end point). This is given by
\[2^{\frac{m-1}{2}}(m+1)\frac{\left(n-\frac{m+3}{2}\right)!}{\left(\frac{n}{2}\right)!\left(\frac{n-m-1}{2}\right)!}.\]

When $n$ is odd and $m$ is even, note that if a given word of length $n$ has prefix and suffix lattice paths which after $\left\lfloor\frac{n}{2}\right\rfloor$ steps have intersected $r$ times, then the word has a total of $2r$ distinct abelian borders. Therefore, for a word with $m$ distinct abelian borders, the number of options for the first $\left\lfloor\frac{n}{2}\right\rfloor$ steps is
\[2^{\frac{m}{2}}{{n-\frac{m}{2}-1}\choose{\frac{n-1}{2}}}.\]
Since the first $\left\lfloor\frac{n}{2}\right\rfloor$ steps determine all but one of the letters in the word, and the middle letter does not affect the number of abelian borders, the total number of binary words of length $n$ with $m$ abelian borders is twice the above quantity.
Finally, when $n$ and $m$ are odd, it follows immediately from the above that an odd number of abelian borders can only occur when the length is even. Therefore, there are no words of odd length having an odd number of distinct abelian borders.
\end{proof}

\subsection{The partial word case}

 We generalize Theorem~\ref{minimalborderfull} and Theorem~\ref{binarywordswithborders} to partial words. First, we count the number of partial words of a given length with a given minimal abelian border length. 

\begin{theorem}
The number of partial words of length $n$ with $h$ holes over an alphabet of size $\sigma$ with minimal abelian border of length ${k}$, $n \geq 2k$, is
\[\sum_{j=0}^{h}g_{\sigma}({k},j){{n-2{k}}\choose{h-j}}{\sigma^{n-2{k}-h+j}},\]
where $g_{\sigma}({k},j)$ is the number of words of length $2{k}$ over an alphabet of size $\sigma$ with $j$ holes and a minimal abelian border of length ${k}$.
\label{generalpwordminimalborder}
\end{theorem}
\begin{proof}
For $0\leq j\leq h$ there are $g_{\sigma}({k},j)$ choices for the prefix and suffix of length ${k}$ containing a total of $j$ holes, and there are $n-2{k}$ remaining positions in the partial word, exactly $h-j$ of which must be holes. For each position that is a letter, there are $\sigma$ possibilities, and there are $n-2{k}-h+j$ such positions.
\end{proof}

In the $k=1$ case,  we can compute the $g_{\sigma}(k,j)$'s. For $h = 1$, the  number of partial words of length $n$ with $h$ holes over an alphabet of size $\sigma$ with minimal abelian border of length one, where $n \geq 2$, is \[\sum_{j=0}^{1}g_{\sigma}(1,j){{n-2}\choose{1-j}}{\sigma^{n+j-3}},\] where $g_{\sigma}(1,0)=\sigma$ and $g_{\sigma}(1,1)=2\sigma$. For $h \geq 2$, the  number of partial words of length $n$ with $h$ holes over an alphabet of size $\sigma$ with minimal abelian border of length one, where $n \geq 2$, is \[\sum_{j=0}^{h}g_{\sigma}(1,j){{n-2}\choose{h-j}}{\sigma^{n-h+j-2}},\]
where $g_{\sigma}(1,0)=\sigma$, $g_{\sigma}(1,1)=2\sigma$, $g_{\sigma}(1,2)=1$, and $g_{\sigma}(1,j)=0$ for $j \in [3..h]$.
In the $\sigma=2$ case, we can compute the $g_{\sigma}({k},j)$'s using the next corollary. To illustrate the proof,
the total words over $\{0,1\}$, up to renaming of letters and reversal, of length $2k=10$ with minimal abelian border length $k=5$ are
\begin{center}
$0000\underline{1}\underline{0}\underline{0}\underline{0}\underline{0}1,
000\underline{1}\underline{1}\underline{0}\underline{0}\underline{0}11,
000\underline{1}\underline{1}\underline{0}\underline{0}1\underline{0}1,
000\underline{1}\underline{1}\underline{0}1\underline{0}\underline{0}1,$\\
$0\underline{0}10\underline{1}\underline{0}\underline{0}\underline{0}11, 
00\underline{1}0\underline{1}\underline{0}\underline{0}1\underline{0}1, 
00\underline{1}\underline{1}\underline{1}\underline{0}11\underline{0}1,$
\end{center}
the underlined positions are the $k$ ways to insert a hole while preserving the minimal abelian border.

\begin{corollary}
\label{cor}
For $k \geq 2$ and $0 \leq j \leq k$, the equality $g_2({k},j)=2\binom{k}{j}C_{k-1}$ holds. Thus for $k \geq 2$, the number of binary partial words with $h$ holes of length $n$, $n\geq2{k}$, with minimal abelian border of length ${k}$ is \[f(n,{k},h)={2^{n-2{k}+1}}C_{k-1}\sum_{j=0}^{h}\binom{k}{j}{{n-2{k}}\choose{h-j}}{2^{j-h}}.\]
\end{corollary}

Second, we count the number of partial words of a given length with a given number of distinct abelian borders. In the proof of Proposition~\ref{nonborderedpartials}, we saw that a partial word of length $n$ with one hole must have at least an abelian border of length $n-1$. By induction on $h$, if we insert a hole into a partial word with $h-1$ holes, we are guaranteed to get an additional abelian border of length $n-h$. 

\begin{theorem}
A partial word $w$ of length $n$ with $h$ holes, $0\leq h<n$, has at least $h$ distinct abelian borders. In particular, $w$ has at least abelian borders of lengths $n-1,n-2,\dots,n-h$.
\label{atleasthborders}
\end{theorem}

We next focus on the one-hole case. Experimental data suggest some elegant formulas in this case.

\begin{theorem}
\label{oldconj}
If $n$ is even, then the number of binary partial words of length $n$ with exactly one hole and $n-1$ distinct abelian borders is $4(3^{n/2}-2^{n/2})$.
\end{theorem}
\begin{proof}
It is clear that for $n$ even and $m=n-1$, each binary partial word of length $n$ with exactly one hole and $m$ distinct abelian borders  generates, through renaming of letters and reversal, three other partial words with the same properties. So to get the number of partial words satisfying the desired properties, we need to find the number of ``unique'' partial words satisfying them and multiply by $4$; ``unique'' is defined in terms of the hole being in the first half of the partial word and the first non-hole letter being $0$. The first few values are given in Table~\ref{tabconj}. For example, for $n=6$, there are 9 unique partial words with $\d$ in position~$0$, 6 with $\d$ in position~$1$, and 4 with $\d$ in position~$2$, for a total of $19$. 
\begin{table}[h!]
\centering
\begin{tabular}{c | c| c| c| c| c  }
${\text{length}} \backslash {i}$ & 0 & 1 & 2 &  $\cdots$ & $\frac{n-2}{2}$ \\
\hline
2 & 1 & &&&\\
4 & 3 & 2 &&&\\
6 & 9 & 6 & 4 &&\\
8 & 27 & 18 & 12 & &\\
10 & 81 & 54 & 36 &  \\
\vdots & \vdots & \vdots & \vdots & &\\
$n$ & $2^0 \times 3^{(\frac{n}{2}-1)}$ & $2^1 \times 3^{(\frac{n}{2}-2)}$ & $2^2 \times 3^{(\frac{n}{2}-3)}$ & $\cdots$ & $2^{(\frac{n}{2}-1)} \times 3^0$
\end{tabular}
\caption{Number of unique binary partial words of even length $n$ with one hole in position~$i$ and with $n-1$ distinct abelian borders.}
\label{tabconj}
\end{table}

Thus, the number of unique partial words $N$ for a given $n$ is 
\begin{center}
$N = 2^0 \times 3^{(\frac{n}{2}-1)} + 2^1 \times 3^{(\frac{n}{2}-2)} + 2^2 \times 3^{(\frac{n}{2}-3)} + \cdots + 2^{(\frac{n}{2}-1)} \times 3^0.$
\end{center}
This is a geometric series with ratio $\frac{2}{3}$, so $N=3^{n/2}-2^{n/2}$. 

For $n=2$, there is one unique partial word. For $n=4$, there are three unique partial words with ${\d}$ in position~$0$ and two with $\d$ in position~$1$. Then to build the partial words of length $n+2$ from those of length $n$, when $n \geq 4$, we proceed as follows. Each partial word $a_0 \cdots a_{\frac{n-2}{2}}a_{\frac{n}{2}} \cdots a_{n-1}$ of length $n$ gives rise to a partial word of length $n+2$ by inserting a pair of suitable letters between $a_{\frac{n-2}{2}}$ and $a_{\frac{n}{2}}$, to obtain three partial words of the form  $a_0 \cdots a_{\frac{n-2}{2}}aba_{\frac{n}{2}} \cdots a_{n-1}$. Then there are $2^{\frac{n}{2}}$ additional unique partial words of length $n+2$ with $\d$ in position~$\frac{n}{2}$.
\end{proof}

It is clear that each binary partial word of even length $n$ with exactly one hole having a minimal abelian border of length $n-1$ generates, through renaming of letters and reversal, three other partial words with the same properties. So to get the number of partial words satisfying the desired properties, we need to find the number of ``unique'' partial words satisfying them and multiply by $4$; ``unique'' is defined in terms of the hole being in the first half of the partial word and the first non-hole letter being $0$. 

\begin{lemma}
\label{lem5m1}
If a binary partial word of even length $n$, $n \geq 4$, with exactly one hole in the first half having a minimal abelian border of length $n-1$ starts with $0$, then it ends with $11$. Therefore if such $w=uv$, with $|u|=|v|$, is unique, then $|u|_0 > |v|_0$ and $|u|_1 < |v|_1-1$.
\end{lemma}

\begin{lemma}
\label{lem1m1}
The number of unique binary partial words of even length $n$, $n \geq 4$, with exactly one hole in the first half and no 1 in the first half having a minimal abelian border of length $n-1$ is $(\frac{n}{2}-1)2^{(\frac{n}{2}-2)}$.
\end{lemma}

\begin{lemma}
\label{lem4m1}
The number of unique binary partial words of even length $n$, $n \geq 4$, with exactly one hole in the first half and exactly one 0 in the first half having a minimal abelian border of length $n-1$ is $1$; such partial word is of the form $0{\d}1^{n-2}$.
\end{lemma}

\begin{proposition}
\label{lem3m1}
The number of unique binary partial words of even length $n$, $n \geq 6$, with exactly one hole in the first half and exactly two $0$'s in the first half having a minimal abelian border of length $n-1$ is $\frac{n^2}{4}-\frac{n}{2}-2$.
\end{proposition}
\begin{proof}
We proceed by induction. For the basis $n=6$, the unique partial words are $0{\d}0011$, $0{\d}0111$, $00{\d}011$, $00{\d}111$ and the result holds. For the inductive step, we show how to build the unique partial words $w'=a_0 \cdots a_{\frac{n}{2}-1}aba_{\frac{n}{2}} \cdots a_{n-1}$ of length $n+2$ with minimal abelian border of length $(n+2)-1$ and exactly two $0$'s in the first half, where $ab \in \{00, 01, 10, 11\}$, from the unique partial words $w=a_0 \cdots a_{\frac{n}{2}-1}a_{\frac{n}{2}} \cdots a_{n-1}$ of length $n$ with minimal abelian border of length $n-1$ but not necessarily with exactly two $0$'s in the first half (there could be only one $0$ in the first half). 

By Lemma~\ref{lem5m1}, the first half of $w$, say $u$, has more $0$'s than the second half, say $v$. If we start with a unique partial word $w$ with $|u|_0=2$, then we obtain a unique partial word $w'$ that satisfies our desired properties by inserting the pair $ab=11$ of letters between $a_{\frac{n}{2}-1}$ and $a_{\frac{n}{2}}$. There are $\frac{n^2}{4}-\frac{n}{2}-2$ such $w$'s by the inductive hypothesis. For some of these $w$'s, it is also possible to insert the pair $ab=10$, i.e., when $|v|_0=0$. There are $n-4$ such $w$'s resulting in $n-4$ unique $w'$ with $10$ in the middle.  Indeed when $n=6$, the $n-4=2$ unique $w'$ are $0{\d}010111$ and $00{\d}10111$. When $n>6$, we first claim that $w'[\frac{n}{2}+2..n+1]=1^{\frac{n}{2}}$. Indeed by Lemma~\ref{lem5m1}, the number of $0$'s in $w'[\frac{n}{2}+1..n+1]$ is zero or one, so it is one in this case. We now claim that $w'[1]\not = 1$; otherwise, $w'$ has an abelian border of length $(n+2)-2<(n+2)-1$. If $w'[0..1]=0{\d}$, then there are $\frac{n}{2}-2$ positions left in the first half to position the other $0$, while if $w'[0..1]=00$, then there are $\frac{n}{2}-2$ positions left in the first half to position the $\d$. 

If we start with a unique partial word $w$ with $|u|_0=1$, then $w=0{\d}1^{n-2}$ by Lemma~\ref{lem4m1}. We obtain a unique partial word $w'$ that satisfies our desired properties by inserting $ab=00$ or $ab=01$.  

There are also two unique $w'$ with a $\d$ in position~$\frac{n}{2}$. When $n=6$, the two unique $w'$ are $001{\d}0111$ and $001{\d}1111$. When $n>6$, we first claim that $w'[\frac{n}{2}+2..n+1]=1^{\frac{n}{2}}$. Indeed by Lemma~\ref{lem5m1}, the number of $0$'s in $w'[\frac{n}{2}+1..n+1]$ is zero or one. If it is one and that $0$ appears in $w'[\frac{n}{2}+2..n+1]$, then $w'$ has an abelian border of length $\frac{n}{2}+2<(n+2)-1$, implying a contradiction. We now claim that $w'[1]=0$; otherwise, $w'$ has an abelian border of length $(n+2)-2<(n+2)-1$. So the two unique $w'$ are $001^{\frac{n}{2}-2}{\d}01^{\frac{n}{2}}$ and $001^{\frac{n}{2}-2}{\d}11^{\frac{n}{2}}$. Altogether, there are $\frac{n^2}{4}-\frac{n}{2}-2 +n-4 + 2 +2=\frac{(n+2)^2}{4}-\frac{n+2}{2}-2$ unique $w'$.
\end{proof}

\begin{proposition}
\label{lem2m1}
The number of unique binary partial words of even length $n$, $n \geq 6$, with exactly one hole in the first half and exactly one 1 in the first half having a minimal abelian border of length $n-1$ is $(\frac{n}{2}-2)(2^{(\frac{n}{2}-2)}(\frac{n}{2}-3)+1)$.
\end{proposition}
\begin{proof}
We proceed by induction. For the basis $n=6$, the unique partial word is $0{\d}1111$ by Lemma~\ref{lem4m1} and the result holds. For the inductive step, we show how to build the unique partial words $w'=a_0 \cdots a_{\frac{n}{2}-1}aba_{\frac{n}{2}} \cdots a_{n-1}$ of length $n+2$ with minimal abelian border of length $(n+2)-1$ and exactly one $1$ in the first half, where $ab \in \{00, 01, 10, 11\}$, from the unique partial words $w=a_0 \cdots a_{\frac{n}{2}-1}a_{\frac{n}{2}} \cdots a_{n-1}$ of length $n$ with minimal abelian border of length $n-1$ but not necessarily with exactly one $1$ in the first half (there could be no $1$ in the first half). 

Let $u$ be the first half of $w$ and $v$ be the second half. If we start with a unique partial word $w$ with $|u|_1=1$, then we obtain a unique partial word $w'$ that satisfies our desired properties by inserting the pair $ab=00$ or $ab=01$ of letters between $a_{\frac{n}{2}-1}$ and $a_{\frac{n}{2}}$. There are $(\frac{n}{2}-2)(2^{(\frac{n}{2}-2)}(\frac{n}{2}-3)+1)$ such $w$'s by the inductive hypothesis.  If we start with a unique partial word $w$ with $|u|_1=0$, then there are $(\frac{n}{2}-1)2^{(\frac{n}{2}-2)}$ such $w$'s by Lemma~\ref{lem1m1}. We obtain a unique partial word $w'$ that satisfies our desired properties by inserting $ab=11$. For some of those $w$'s, we can also insert $ab=10$. To see how many such $w'$'s there are, the $\d$ can be in any position in $\{1, \ldots, \frac{n}{2}-1\}$. The last two positions are $11$. There are $2^{(\frac{n}{2}-2)}-1$ choices for the $\frac{n}{2}-2$ positions after the middle $10$. The $-1$ comes from the fact that these positions cannot be all $0$'s (otherwise, there would be an abelian border of length $\frac{n}{2}+1<(n+2)-1$).

There are also $(\frac{n}{2}-3)2^{(\frac{n}{2}-1)}+2$ unique $w'$ with a $\d$ in position~$\frac{n}{2}$.  We first claim that $w'[1]=0$; otherwise, $w'$ has an abelian border of length $(n+2)-2<(n+2)-1$. We next claim that for $2 \leq i < \frac{n}{2}$, if $w'[0..i]=0^i1$ then $w'[n-i+1..n-1] \not = 0^{i-1}$; otherwise, $w'$ has an abelian border of length $(n+2)-(i+1) < (n+2)-1$. So if the $1$ in the first half is in position~$2$, there are $(2^1-1)2^{(\frac{n}{2}-2)}$ unique corresponding $w'$, if the $1$ in the first half is in position~$3$, there are $(2^2-1)2^{(\frac{n}{2}-3)}$ unique corresponding $w'$, \ldots, and if the $1$ in the first half is in position~$\frac{n}{2}-1$, there are $(2^{(\frac{n}{2}-2)}-1)2^{1}$ unique corresponding $w'$. So, we get a total of $(\frac{n}{2}-3)2^{(\frac{n}{2}-1)}+2$ unique $w'$ with a $\d$ in position~$\frac{n}{2}$. 
\end{proof}

\begin{theorem}
\label{conjm1}
The number of binary partial words of even length $n$ with exactly one hole having a minimal abelian border of length $n-1$ is $\Theta(2^{n-1})$.
\end{theorem}
\begin{proof}
Let $f(n,n-1,1)$ denote the number of binary partial words of even length $n$ with exactly one hole having a minimal abelian border of length $n-1$. For $n=2$, there is one unique partial word with $\d$ in position~$0$. For $n=4$, there is one unique partial word with ${\d}$ in position~$1$ by Lemma~\ref{lem1m1} or Lemma~\ref{lem4m1}. Then, for $n \geq 4$, to build the unique partial words of length $n+2$ with minimal abelian border of length $(n+2)-1$ from those of length $n$ with minimal abelian border of length $n-1$, we proceed as follows. Each unique binary partial word $w=a_0 \cdots a_{\frac{n}{2}-1}a_{\frac{n}{2}} \cdots a_{n-1}$ of even length $n$ with exactly one hole having a minimal abelian border of length $n-1$ gives rise to a partial word of length $n+2$ by inserting a pair $ab$ of letters, $ab \in \{00, 01, 11\}$, between $a_{\frac{n}{2}-1}$ and $a_{\frac{n}{2}}$, to obtain three unique partial words of the form  $w'=a_0 \cdots a_{\frac{n}{2}-1}aba_{\frac{n}{2}} \cdots a_{n-1}$. For some $w$'s, it is sometimes also possible to insert the pair $ab=10$.

To see this, let $w$ be as above. By our convention, $w$ starts with $0$, and by Lemma~\ref{lem5m1}, $w$ ends with $11$. Since $w$ has no abelian border of any length in $\{1, \ldots, \frac{n}{2}\}$, $w'$ has no such abelian border.  Clearly,  $w'$ has also no abelian border of any length in $\{\frac{n}{2}+2, \ldots, n\}$. If $ab \in \{00, 11\}$, then $a_0 \cdots a_{\frac{n}{2}-1}a$ is not abelian compatible with $ba_{\frac{n}{2}} \cdots a_{n-1}$ and such $w'$ has no abelian border of length $\frac{n}{2}+1$. If $ab=01$, then $a_0 \cdots a_{\frac{n}{2}-1}a$ is not abelian compatible with $ba_{\frac{n}{2}} \cdots a_{n-1}$ since the number of $0$'s in $a_0 \cdots a_{\frac{n}{2}-1}$ is greater than the number of $0$'s in $a_{\frac{n}{2}} \cdots a_{n-1}$ by Lemma~\ref{lem5m1}, and so such $w'$ has no abelian border of length $\frac{n}{2}+1$. 

Then there are $2^{(\frac{n}{2}-2)}+2^{n-4}$ additional unique partial words $w'$ of length $n+2$ with $\d$ in position~$\frac{n}{2}$. To see this, we first claim that $w'[1]=0$; otherwise, $w'$ has an abelian border of length $(n+2)-2 < (n+2)-1$. So we can write $w'=00u{\d}v11$, where $u$ and $v$ are total words of length $\frac{n}{2}-2$ and $\frac{n}{2}-1$, respectively. There are $2^{\frac{n}{2}-2}$ possibilities for $u$, ranging from $u=1^{\frac{n}{2}-2}$ to $u=0^{\frac{n}{2}-2}$. On the one hand, by the proof of Proposition~\ref{lem3m1}, if $u=1^{\frac{n}{2}-2}$ then there are 2 unique partial words of length $n+2$ with $\d$ in position~${\frac{n}{2}}$, i.,e., $001^{\frac{n}{2}-2}{\d}01^{\frac{n}{2}}$ and $001^{\frac{n}{2}-2}{\d}11^{\frac{n}{2}}$. On the other hand if $u=0^{\frac{n}{2}-2}$, there are $2^{\frac{n}{2}-1}$ unique such partial words. Altogether, there are $2+4+6+ \cdots + 2^{(\frac{n}{2}-1)}$ unique partial words of length $n+2$ with $\d$ in position~$\frac{n}{2}$, for a total of $2(1+2+3+ \cdots + 2^{(\frac{n}{2}-2)})=2^{(\frac{n}{2}-2)}+2^{n-4}$. 

Thus, the number of unique partial words with exactly one hole of length $n+2$ having a minimal abelian border of length $(n+2)-1$ satisfies  
\begin{center}
$\frac{3f(n,n-1,1)}{4}+2^{(\frac{n}{2}-2)}+2^{n-4} < \frac{f(n+2,(n+2)-1,1)}{4} < \frac{4f(n,n-1,1)}{4}+2^{(\frac{n}{2}-2)}+2^{n-4}.$
\end{center}
From this we deduce that, for $n$ sufficiently large, if $f(n,n-1,1)$ is $O(2^{n-1})$ then  $f(n+2,(n+2)-1,1)$ is $O(2^{(n+2)-1})$, and if $f(n,n-1,1)$ is $\Omega(2^{n-1})$ then  $f(n+2,(n+2)-1,1)$ is $\Omega(2^{(n+2)-1})$. 
\end{proof}

Theorem~\ref{conjm1} equivalently states that the number of binary partial words of even length $n$ with exactly one hole and one abelian border is $\Theta(2^{n-1})$. This is implied by Theorem~\ref{atleasthborders} that guarantees at least a minimal border of length $n-1$ when a partial word has at least one hole. We can prove that for $n$ sufficiently large, the number of binary partial words of even length $n$ with exactly one hole and $m$, $m>1$ odd, distinct abelian borders is bounded from above by the number of binary partial words of length $n$ with exactly one hole and one abelian border. The proof is based on the following lemmas.

First, we consider unique binary partial words of even length $n$ with exactly one hole in position~$0$ and $m$ distinct abelian borders.

\begin{lemma}
\label{lem1m}
Let $m >1$ be odd. If a binary partial word of even length $n$, $n \geq 4$, has exactly one hole in position~$0$ and $m$ distinct abelian borders, then three of these abelian borders have lengths $1, \frac{n}{2}, n-1$. Furthermore, for $k \in \{1, \ldots, \frac{n}{2}-1\}$, we have that $k$ is an abelian border length if and only if $n-k$ is an abelian border length.
\end{lemma}

Similarly, we can show that for odd $m$, if a binary partial word of odd length $n$, $n \geq 3$, has exactly one hole and $m$ distinct abelian borders, then the hole is not in position~$0$.

\begin{lemma}
\label{lem3m}
Let $m >1$ be odd. If a partial word $w=a_0 \cdots a_{\frac{n}{2}-1}a_{\frac{n}{2}} \cdots a_{n-1}$ of even length $n$, $n \geq 6$, over the alphabet $\{0,1\}$ has exactly one hole in position~$0$ and $m$ distinct abelian borders, then for $j \in \{2, \ldots, \frac{n}{2}-1\}$ and $ab \in \{00, 11\}$, the partial word 
$w'=a_0 \cdots a_{j-1} a a_{j} \cdots a_{\frac{n}{2}-1} a_{\frac{n}{2}} \cdots a_{n-j-1} b a_{n-j} \cdots a_{n-1}$ 
has exactly $m$ distinct abelian borders.
\end{lemma}

Now, we consider unique binary partial words of even length $n$ with exactly one hole in position~$1$ and $m$ distinct abelian borders. 

\begin{lemma}
\label{lem456m}
Let $m$ be odd. If a binary partial word of even length $n$, $n \geq 6$, has exactly one hole in position~$1$ and $m$ distinct abelian borders, then $n-1$ is an abelian border length and for $k \in \{2, \ldots, \frac{n}{2}-1\}$ we have that $k$ is an abelian border length if and only if $n-k$ is an abelian border length. Consequently, either both $1, \frac{n}{2}$ are abelian border lengths or none is an abelian border length. 
\end{lemma}

If $m=3$ in Lemma~\ref{lem456m}, then an abelian border of length $1$ would imply an abelian border of length 2 (and also one of length $n-2$), so there would be at least four distinct border lengths, a contradiction. So neither $1$ nor $\frac{n}{2}$ is an abelian border length in the $m=3$ case when the position of the hole is $1$.

\begin{lemma}
\label{lem456mf}
Let $m$ be odd. If a partial word $w=a_0 \cdots a_{\frac{n}{2}-1}a_{\frac{n}{2}} \cdots a_{n-1}$ of even length $n$, $n \geq 6$, over $\{0,1\}$ has exactly one hole in position~$1$ and $m$ distinct abelian borders, then for at least three pairs $ab$ in $\{00, 01, 10, 11\}$, the partial word 
$w'=a_0 \cdots a_{\frac{n}{2}-1} ab a_{\frac{n}{2}} \cdots a_{n-1}$
has exactly $m$ distinct abelian borders.
\end{lemma}

Next, we consider unique binary partial words of even length $n$ with exactly one hole in position~$i$, where $2 \leq i \leq \frac{n}{2}-1$, and $m$ distinct abelian borders.

\begin{lemma}
\label{lem78m}
Let $m>1$ be odd. If a binary partial word of even length $n$, $n \geq 8$, has exactly one hole in position~$i$, $2 \leq i \leq \frac{n}{2}-1$,  and $m$ distinct abelian borders, one whose length is 1, then two of these abelian borders have lengths $n-2, n-1$. Furthermore, for $k \in \{i+1, \ldots, \frac{n}{2}-1\}$, we have that $k$ is an abelian border length if and only if $n-k$ is an abelian border length.
\end{lemma}

\begin{lemma}
\label{lem78mf}
Let $m>1$ be odd. If a partial word $w=a_0 \cdots a_{\frac{n}{2}-1}a_{\frac{n}{2}} \cdots a_{n-1}$ of even length $n$, $n \geq 8$, over the alphabet $\{0,1\}$ has exactly one hole in position~$i$, $2 \leq i \leq \frac{n}{2}-1$, and $m$ distinct abelian borders, one whose length is $1$, then for at least three pairs $ab$ in $\{00, 01, 10, 11\}$, the partial word 
$w'=a_0 \cdots a_{\frac{n}{2}-1} ab a_{\frac{n}{2}} \cdots a_{n-1}$
has exactly $m$ distinct abelian borders.
\end{lemma}

The $m=1$ case of the next theorem follows from Theorem~\ref{conjm1}. The general $m>1$ odd case is similar to the $m=1$ case, the only difference is that we can have some more trivial values. We can prove the result based on the location of the hole, inserting three or four pairs in $\{00, 01, 10, 11\}$ in suitable positions of partial words of even length $n$ to get partial words of length $n+2$, and finding the dominant term. 

\begin{theorem}
\label{conjm}
If $m$ is odd, then the number of binary partial words of length $n$ with exactly one hole and $m$ distinct abelian borders is $\Theta(2^{n-1})$.
\end{theorem}

\section{Conclusion}

Computer experiments strongly suggest that the following conjectures are true. The integer $\texttt{A191386}(n)$ in Sloane's Online Encyclopedia of Integer Sequences (OEIS) gives the number of ascents of length 1 in all dispersed Dyck paths of length $n$. A \emph{dispersed Dyck path} is one or more Dyck paths connected by one or more horizontal steps \cite{ChGu}, or equivalently a Motzkin path of length $n$ with no horizontal steps at positive heights. An \emph{ascent} is a maximal sequence of upward steps. The first few values of this sequence (starting at $n=0$) are 0, 0, 1, 2, 5, 10, 23, 46, 102, 204, 443, 886, 1898, 3796, 8054. 

\begin{conjecture}
The number of binary partial words of length $n$ with exactly one hole having a minimal abelian border of length ${k}<n$ is
\[f(n,{k},1)=
\begin{cases}
\texttt{A191386}(n), & \text{if }n\geq3\text{ is odd and }{k}=n-1; \\
2\texttt{A191386}(n-1), & \text{if }n\geq4\text{ is even and }{k}=n-1; \\
2\texttt{A191386}(n-2), & \text{if }n\geq5\text{ is odd and }{k}=n-2;\\
4\texttt{A191386}(n-3), & \text{if }n\geq6\text{ is even and }{k}=n-2.\\ 
\end{cases}\]
\label{sabdiagonal}
\end{conjecture}

Note that Theorem~\ref{atleasthborders} suggests that Conjecture~\ref{sabdiagonal} is true if and only if the number of binary partial words of length $n$ with exactly one hole and exactly one abelian border is $\texttt{A191386}(n)$ when $n$ is odd and $2\texttt{A191386}(n-1)$ when $n$ is even.

\begin{conjecture}
The set of binary total words without abelian borders is non-context-free.
\end{conjecture}

\section*{Acknowledgements}

Project sponsored by the National Security Agency under Grant Number H98230-15-1-0232. The United States Government is authorized to reproduce and distribute reprints notwithstanding any copyright notation herein. This manuscript is submitted for publication with the understanding that the United States Government is authorized to reproduce and distribute reprints. This material is based upon work supported by the National Science Foundation under Grant No. DMS--1060775. We thank Benjamin De Winkle for his very useful comments and suggestions. We also thank the referees of preliminary versions of this paper for their very valuable comments and suggestions.

\nocite{*}
\bibliographystyle{eptcs}
\bibliography{proba}

\end{document}